\documentclass[conference, a4paper]{IEEEtran}

\renewcommand\IEEEkeywordsname{Index Terms}

\usepackage{graphicx}
\usepackage{caption}
\usepackage{cite}
\usepackage[hyphens]{url}
\usepackage[nolist]{acronym} 
\usepackage{pgfplots}
\usepgfplotslibrary{groupplots}

\usetikzlibrary{arrows,shapes,graphs,graphs.standard,quotes,arrows.meta,decorations.markings,positioning,calc,patterns,patterns.meta}
\pgfplotsset{compat=newest}
\usepackage[bookmarks=false]{hyperref}
\usepackage{amsmath, amsbsy, amssymb}
\usepackage{amsthm}
\usepackage{multirow}
\usepackage{nicematrix}
\usepackage{booktabs}
\usepackage{nicefrac}

\usepackage{pifont}%
\newcommand{\cmark}{\ding{51}}%
\newcommand{\xmark}{\ding{55}}%

\usepackage[linesnumbered,vlined, ruled]{algorithm2e}

\usepackage{etoolbox}
\makeatletter
\patchcmd\algocf@Vline{\vrule}{\vrule \kern-0.4pt}{}{}
\patchcmd\algocf@Vsline{\vrule}{\vrule \kern-0.4pt}{}{}
\makeatother

\usepackage{marginnote}
\tikzset{>=latex}
\SetAlCapNameFnt{\footnotesize}
\SetAlCapFnt{\footnotesize}
\SetKwProg{Fn}{subroutine}{:}{end}
\let\oldnl\nl%
\newcommand{\nonl}{\renewcommand{\nl}{\let\nl\oldnl}}%
\newcommand{\algrule}[1][.2pt]{\par\vskip.5\baselineskip\hrule height #1\par\vskip.5\baselineskip}
\SetKwBlock{Parallel}{parallel:}{}
\SetKwProg{Task}{Task}{:}{end}
\SetKw{KwSet}{set}
\SetKw{KwAwait}{await}
\SetKw{KwCancel}{cancel}
\SetInd{0.4em}{0.8em}

\definecolor{mblue}{RGB}{0, 126, 198}
\definecolor{ured}{RGB}{238, 28 35}
\definecolor{green}{RGB}{140, 198, 62}
\definecolor{dgreen}{RGB}{34,139,34}
\definecolor{yellow}{RGB}{255, 229, 0}
\definecolor{orange}{RGB}{244, 111, 33}
\definecolor{pink}{RGB}{237, 0, 140}
\definecolor{purple}{RGB}{128, 10, 145}
\definecolor{lgray}{RGB}{224, 224, 224}
\definecolor{mgray}{RGB}{128, 128, 128}
\definecolor{dgray}{RGB}{80,80,80}

\IEEEoverridecommandlockouts

\renewcommand{\vec}[1]{\boldsymbol{#1}}
\newcommand{\vecs}[1]{\boldsymbol{#1}}

\newcommand{\cv}{\vec{c}}

\newcommand{\mv}{\vec{m}}

\newcommand{\rv}{\vec{r}}

\newcommand{\tv}{\vec{t}}
\newcommand{\uv}{\vec{u}}
\newcommand{\vv}{\vec{v}}

\newcommand{\xv}{\vec{x}}
\newcommand{\yv}{\vec{y}}

\newcommand{\zerov}{\vec{0}}
\newcommand{\onev}{\vec{1}}
\newcommand{\ellv}{\vecs{\ell}}

\newcommand{\Gm}{\vec{G}}

\newcommand{\Cc}{{\cal C}}
\newcommand{\Dc}{{\cal D}}

\newcommand{\FF}{\mathbb{F}}

\newcommand{\argmax}{\mathop{\mathrm{argmax}}}

\newtheorem{theorem}{Theorem}

\newtheorem{lemma}{Lemma}

\begin{document}

\begin{NoHyper}
\title{Improved Acceptance Criteria for Speculative Successive Cancellation Decoding of Polar Codes}

\author{\IEEEauthorblockN{Marvin Rübenacke, Ryan Seah, and Warren J. Gross\\}
	\IEEEauthorblockA{
		Department of Electrical and Computer Engineering, McGill University, Montréal, Québec, Canada
		\\\{marvin.ruebenacke, warren.gross\}@mcgill.ca, ryan.seah@mail.mcgill.ca\\
	}
  }

\maketitle

\begin{acronym}
\acro{ML}{maximum likelihood}
\acro{BP}{belief propagation}
\acro{LDPC}{low-density parity-check}
\acro{BER}{bit error rate}
\acro{SNR}{signal-to-noise-ratio}
\acro{BPSK}{binary phase shift keying}
\acro{BI-AWGN}{binary-input additive white Gaussian noise}
\acro{AWGN}{additive white Gaussian noise}
\acro{LLR}{log-likelihood ratio}
\acro{MAP}{maximum a posteriori}
\acro{FER}{frame error rate}
\acro{BLER}{block error rate}
\acro{SCL}{successive cancellation list}
\acro{SC}{successive cancellation}
\acro{BI-DMC}{Binary Input Discrete Memoryless Channel}
\acro{CRC}{cyclic redundancy check}
\acro{CA-SCL}{CRC-aided successive cancellation list}
\acro{BEC}{Binary Erasure Channel}
\acro{BSC}{Binary Symmetric Channel}
\acro{RM}{Reed--Muller}
\acro{SISO}{soft-in/soft-out}
\acro{3GPP}{3rd Generation Partnership Project}
\acro{eMBB}{enhanced Mobile Broadband}
\acro{CN}{check node}
\acro{VN}{variable node}
\acro{URLLC}{ultra-reliable low-latency communications}
\acro{mMTC}{massive machine-type communications}
\acro{DMC}{discrete memoryless channel}
\acro{QC}{quasi-cyclic}
\acro{5G}{fifth generation mobile telecommunication}
\acro{SCAN}{soft cancellation}
\acro{LSB}{least significant bit}
\acro{MSB}{most significant bit}
\acro{PPV}{Polyanskyi-Poor-Verd\'{u}}
\acro{Spec-SC}{speculative successive cancellation}
\acro{SSC}{simplified SC}
\acro{FSSC}{fast simplified SC}
\acro{HD}{Hamming distance}
\acro{ED}{ellipsoidal distance}
\acro{SPC}{single parity-check}
\acro{REP}{repetition}
\end{acronym}

\begin{abstract}
Next-generation data-channel applications of polar codes demand decoding algorithms with high throughput and low latency.
The recently proposed \ac{Spec-SC} decoding reduces average decoding latency by speculatively executing the right branch of \ac{SC} decoding in parallel with the left branch, verifying the result once the g-function is computed.
In this paper, we propose new acceptance criteria based on Hamming distance and ellipsoidal distance that improve acceptance rates, enabling greater speed-ups for lower-rate nodes.
We further show that combining code membership testing with speculative decoding accelerates high-rate nodes as well.
Numerical results confirm that both approaches outperform \ac{FSSC} decoding and \ac{Spec-SC} with the original acceptance condition in terms latency by a wide margin, with virtually identical error-rate performance.

\end{abstract}
\acresetall

\section{Introduction}
Polar codes, introduced by Ar{\i}kan in 2009, are the first class of error-correcting codes proven to achieve the capacity of binary-input memoryless symmetric channels under \ac{SC} decoding \cite{ArikanMain}.
Already standardized for the 5G control channel, polar codes are becoming an attractive candidate for data channel applications of next-generation communication systems \cite{proceedings2024trends}.
Compared to the iterative decoding of \ac{LDPC} codes, \ac{SC} decoding exhibits substantially lower complexity, making polar codes appealing where computational efficiency is critical \cite{krieg2025longpolarvsldpc, huawei2021fastpolarcodesterabitspersecond}.
However, the sequential nature of \ac{SC} decoding limits throughput, since it scales inversely with decoding latency.

Several \ac{SC} variants address this limitation. \Ac{SSC} decoding collapses parts of the decoding tree by identifying rate-0 and rate-1 component codes \cite{simplfiedSC}, while \ac{FSSC} decoding recognizes additional node types that can be decoded directly, further reducing complexity and latency \cite{sarkis2014fast, hanif2017newnodes}. Syndrome-check \ac{SC} decoding \cite{yoo2016checkconstituent, choi2017recusivesyndrome, kim2018earlystopping} dynamically skips recursion when a subtree input is already a valid codeword. More recently, adaptive tree-pruning techniques based on channel quality \cite{singh2025ribrpruning, gunturu2025adaptivepruning} and per-node hard-decision decoding under good channel conditions \cite{sun2024fastdecdoing} have also been proposed.

\Ac{BP} decoding \cite{ArikanBP, ArikanBP_original} offers an alternative that updates messages in parallel, processing left and right branches simultaneously. However, matching the error-correction performance of \ac{SC}-based decoders typically demands many iterations, incurring high complexity. However, early stopping criteria can reduce average latency without sacrificing performance.

\Ac{Spec-SC} decoding~\cite{seah2026specsc} combines the benefits of parallel processing in \ac{BP} decoding with the low complexity of \ac{SC} decoding by speculatively executing the right branch concurrently with the left branch. A verification criterion is used to trigger recomputation when needed, thus preserving the error-correction performance of conventional \ac{SC} decoding. 
Its achievable speedup depends primarily on this verification criterion. In the original formulation~\cite{seah2026specsc}, the speculative right-branch result is accepted only if it exactly matches the hard-decision on  the $g$-function output. While this approach is computationally efficient, it frequently rejects correct speculative decoding results, thereby limiting the achievable decoding speedup.

In this paper, we analyze and improve the \ac{Spec-SC} acceptance criterion. Our contributions are:
\begin{enumerate}
\item We propose accepting the speculative right-branch result when it is sufficiently close (rather than identical) to the $g$-function output, using hard-decision (e.g., Hamming distance) or soft-decision (e.g., ellipsoidal distance) metrics to better capture the right branch's error-correction capability. This yields significantly larger speed-ups, particularly for low-to-medium rate nodes.
\item By decomposing the original verification criterion into distinct cases, we show it triggers only in a subset of the scenarios covered by the syndrome-based pruning method of \cite{yoo2016checkconstituent}. Building on this insight, we introduce a code membership test that further increases the speed-ups of \ac{Spec-SC} for high-rate nodes.
\end{enumerate}

\section{Preliminaries}

\subsection{Polar Codes}
Polar codes are constructed from the polar transform defined by the $N\times N$ polarization matrix $\boldsymbol G_N = \left[\begin{smallmatrix} 1 & 0 \\ 1 & 1 \end{smallmatrix}\right]^{\otimes n}$ with $N=2^n$, where $(\cdot)^{\otimes n}$ denotes the $n$-th Kronecker power.
The polar transform converts $N$ identical channels into polarized synthetic channels that become either highly reliable or highly unreliable as $N$ increases.
The design for an $(N,K)$ code consists of selecting the $K$ most reliable synthetic channels to form the information set, while the remaining channels constitute the frozen set.
This selection is represented by the rate profile $\rv \in \FF_2^N$, where $r_i = 1$ indicates an information channel and $r_i = 0$ a frozen channel.
Encoding is performed by embedding the message vector $\mv \in \FF_2^K$ into the uncoded vector $\uv \in \FF_2^N$ at the positions indexed by $r_i = 1$, while the remaining positions are set to zero. The resulting vector $\uv$ is then transformed using the polar transform to produce the codeword $\cv = \uv \Gm_N$.

\subsection{Successive Cancellation Decoding}
\Ac{SC} decoding, originally proposed for \ac{RM} codes by Schnabl and Bossert in 1995 \cite{schnabl1995}, is a recursive soft-decision decoding algorithm for polar-like codes.
It computes the most likely value of each uncoded bit successively, i.e.,
\begin{equation}
	\hat{u}_i = \argmax_{u_i} \Pr\left(u_i \mid \hat{u}_0, \dots, \hat{u}_{i-1}, \yv \right).
\end{equation}
Algorithm~\ref{alg:fssc} lists the \ac{FSSC} variant of \ac{SC} decoding in \ac{LLR} domain, where the $f$-function and $g$-function used to compute the left and right branches of the factor tree, respectively, are given by 
\begin{align}
    f(a,b) &=  \operatorname{sgn}(a)\operatorname{sgn}(b) \min(|a|,\,|b|), \\
    g(a,b,c) &= (-1)^{c}\cdot a + b. %
\end{align}
Recursion is terminated once one of the easy-to-decode component codes $\Dc$, such as Rate-0, Rate-1, \ac{REP} or \ac{SPC} nodes are reached \cite{sarkis2014fast, hanif2017newnodes}.

\begin{algorithm}[tbh]
	\caption{Fast simplified successive cancellation decoding.}\label{alg:fssc}
	\SetAlgoVlined
	\LinesNumbered
	\SetKwInOut{Input}{Input}\SetKwInOut{Output}{Output}
	\Input{Channel LLR vector $\ellv$, rate profile $\rv$, set of leaf decoders $\Dc$.}
	\Output{Codeword estimate $ \hat{\cv} $.}
	$\hat{\cv} = \operatorname{FSSC}(\ellv,\,\rv)$\;
	\algrule[.5pt]
    \Fn{$\operatorname{FSSC}(\ellv,\,\rv)$} {
		$N \gets \operatorname{length}(\ellv)$\;
		\eIf{leaf decoder for $\rv$ in $\Dc$}{
			\Return $\operatorname{decode}(\ellv, \rv)$ \;
		} {
			$\hat{\cv}_0 \gets \operatorname{FSSC}(f(\ellv_{0:\nicefrac{N}{2}},\, \ellv_{\nicefrac{N}{2}:N}),\,\rv_{0:\nicefrac{N}{2}})$\;
			$\hat{\cv}_1 \gets \operatorname{FSSC}(g(\ellv_{0:\nicefrac{N}{2}},\,\ellv_{\nicefrac{N}{2}:N},\, \hat{\cv}_0),\,\rv_{\nicefrac{N}{2}:N})
			$\;
			\Return $(\hat{\cv}_0 \oplus \hat{\cv}_1 \mid \hat{\cv}_1)$\;
		}
	}
\end{algorithm}

\subsection{Speculative Successive Cancellation Decoding}

\begin{algorithm}[tbh]
	\caption{Speculative successive cancellation decoding.}\label{alg:specsc}
	\SetAlgoVlined
	\LinesNumbered
	\SetKwInOut{Input}{Input}\SetKwInOut{Output}{Output}
	\Input{Channel LLR vector $\ellv$, rate profile $\rv$, set of leaf decoders $\Dc$.}
	\Output{Codeword estimate $\hat{\cv}$.}
	$\hat{\cv} = \operatorname{Spec-SC}(\ellv,\,\rv)$\;
	\algrule[.5pt]
	\Fn{$\operatorname{Spec-SC}(\ellv,\,\rv)$} {
		$N \gets \operatorname{length}(\ellv)$\;
		\eIf{leaf decoder for $\rv$ in $\Dc$}{
			\Return $\operatorname{decode}(\ellv, \rv)$ \;
		} {
			\Parallel{
				\Task{\normalfont 1}{
					$\hat{\cv}_0 \gets \operatorname{Spec-SC}(f(\ellv_{0:\nicefrac{N}{2}},\, \ellv_{\nicefrac{N}{2}:N}),\,\rv_{0:\nicefrac{N}{2}})$\;
                    $\ellv_g \gets g(\ellv_{0:\nicefrac{N}{2}},\,\ellv_{\nicefrac{N}{2}:N},\, \hat{\cv}_0)$\;
                    \KwSet flag$_1$\;
					$\hat{\cv}_1 \gets \operatorname{Spec-SC}(\ellv_g,\,\rv_{\nicefrac{N}{2}:N})$\;
                    \KwCancel Task 2\;
                    \Return $(\hat{\cv}_0 \oplus \hat{\cv}_1 \mid \hat{\cv}_1)$\;
				}
                \Task{\normalfont 2}{
					$\tilde{\cv}_1 \gets \operatorname{FSSC}( \ellv_{\nicefrac{N}{2}:N},\,\rv_{\nicefrac{N}{2}:N})$\;
                    \KwAwait ($\text{flag}_1$)\; %
					\If{$\tilde{\cv}_1 = \operatorname{hard}(\ellv_g)$}{
                        \KwCancel Task 1\;
						\Return $(\hat{\cv}_0 \oplus \tilde{\cv}_1 \mid \tilde{\cv}_1)$\;
					}
				}
			}
		}
	}
\end{algorithm}

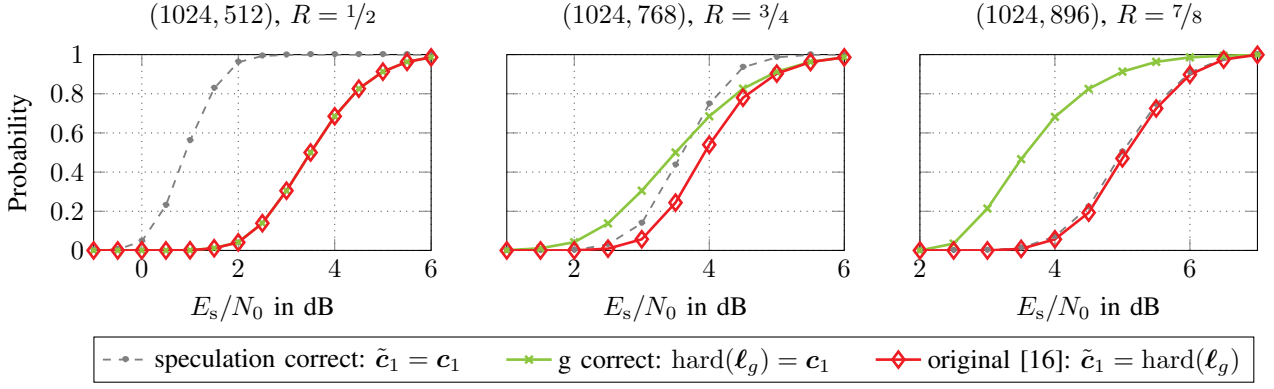
\begin{figure*}[t]
    \centering
    \begin{tikzpicture}
\begin{groupplot}[
    group style={group size=3 by 1, y descriptions at=edge left, horizontal sep=1cm},
    width=\linewidth/3,
    height=0.23\linewidth,
    legend columns=3,
    legend style={/tikz/every even column/.append style={column sep=0.5cm}},
    ylabel={Probability},
    xlabel={$E_\mathrm{s}/N_0$ in dB},
    xmajorgrids,
	yminorticks=true,
	ymajorgrids,
    grid style={dotted,dgray},
    ymin=0,
	ymax=1
]

\nextgroupplot[title={$(1024,512)$, $R=\nicefrac{1}{2}$},
legend to name=polarcoderulez,
xmin=-1,
xmax=6]
\addplot [mgray, line width=0.7pt, dashed, mark=*, mark size=1pt] 
table[col sep=comma]{%
-1.000, 2.000e-03
-0.500, 5.000e-03
0.000, 4.700e-02
0.500, 2.300e-01
1.000, 5.610e-01
1.500, 8.280e-01
2.000, 9.610e-01
2.500, 9.930e-01
3.000, 9.980e-01
3.500, 1.000e+00
4.000, 1.000e+00
4.500, 1.000e+00
5.000, 1.000e+00
5.500, 1.000e+00
6.000, 1.000e+00
};
\addlegendentry{speculation correct: $\Tilde{\cv}_1 = \cv_1$}
\addplot [green, line width=1pt, mark=x] 
table[col sep=comma]{%
-1.000, 0.000
-0.500, 0.000
0.000, 0.000
0.500, 0.000
1.000, 0.001
1.500, 0.012
2.000, 0.042
2.500, 0.139
3.000, 0.305
3.500, 0.500
4.000, 0.685
4.500, 0.826
5.000, 0.913
5.500, 0.963
6.000, 0.986
};
\addlegendentry{g correct: $\operatorname{hard}(\ellv_g) = \cv_1$}

\addplot [ured, line width=1pt, mark=diamond, mark size=3pt] 
table[col sep=comma]{%
-1.000, 0.000
-0.500, 0.000
0.000, 0.000
0.500, 0.000
1.000, 0.001
1.500, 0.011
2.000, 0.041
2.500, 0.139
3.000, 0.305
3.500, 0.500
4.000, 0.685
4.500, 0.826
5.000, 0.913
5.500, 0.963
6.000, 0.986
};
\addlegendentry{original \cite{seah2026specsc}: $\Tilde{\cv}_1 = \operatorname{hard}(\ellv_g)$}

\nextgroupplot[title={$(1024,768)$, $R=\nicefrac{3}{4}$},
xmin=1,
xmax=6]
\addplot [mgray, line width=0.7pt, dashed, mark=*, mark size=1pt] 
table[col sep=comma]{%
1.000, 0.000
1.500, 0.000
2.000, 0.001
2.500, 0.031
3.000, 0.139
3.500, 0.436
4.000, 0.748
4.500, 0.935
5.000, 0.986
5.500, 0.999
6.000, 1.000
};
\addplot [green, line width=1pt, mark=x] 
table[col sep=comma]{%
1.000, 1.000e-03
1.500, 1.100e-02
2.000, 4.200e-02
2.500, 1.380e-01
3.000, 3.050e-01
3.500, 5.000e-01
4.000, 6.850e-01
4.500, 8.260e-01
5.000, 9.130e-01
5.500, 9.630e-01
6.000, 9.860e-01
};

\addplot [ured, line width=1pt, mark=diamond, mark size=3pt] 
table[col sep=comma]{%
1.000, 0.000
1.500, 0.000
2.000, 0.000
2.500, 0.009
3.000, 0.057
3.500, 0.244
4.000, 0.540
4.500, 0.780
5.000, 0.903
5.500, 0.962
6.000, 0.986
};

\nextgroupplot[title={$(1024,896)$, $R=\nicefrac{7}{8}$},
xmin=2,
xmax=7]
\addplot [mgray, line width=0.7pt, dashed, mark=*, mark size=1pt]
table[col sep=comma]{%
2.000, 0.000
2.500, 0.000
3.000, 0.001
3.500, 0.013
4.000, 0.071
4.500, 0.222
5.000, 0.503
5.500, 0.746
6.000, 0.909
6.500, 0.980
7.000, 1.000
};
\addplot [green, line width=1pt, mark=x] 
table[col sep=comma]{%
2.000, 0.000
2.500, 0.035
3.000, 0.214
3.500, 0.466
4.000, 0.682
4.500, 0.825
5.000, 0.913
5.500, 0.963
6.000, 0.986
6.500, 0.994
7.000, 0.999
};

\addplot [ured, line width=1pt, mark=diamond, mark size=3pt] 
table[col sep=comma]{%
2.000, 0.000
2.500, 0.000
3.000, 0.000
3.500, 0.008
4.000, 0.057
4.500, 0.193
5.000, 0.470
5.500, 0.725
6.000, 0.898
6.500, 0.975
7.000, 0.999
};

\end{groupplot}

\node[below=0.1em,inner sep=1pt, xshift=1.6em] at(current bounding box.south) {\pgfplotslegendfromname{polarcoderulez}};
\end{tikzpicture}
    \caption{\footnotesize Probability of the components of the original verification criterion at the root node of \ac{Spec-SC} decoding for the 5G polar code of length $N=1024$.}
    \label{fig:example1024}
\end{figure*}

\Ac{Spec-SC} decoding was proposed in \cite{seah2026specsc} and is summarized in Algorithm~\ref{alg:specsc}.
At each node in the decoding tree, instead of waiting for the result of the first branch to compute the $g$-function, it invokes the right-branch decoding directly on the second half of the \ac{LLR} vector $\ellv$ to be processed in parallel to the left-branch decoding.
As the speculative right branch does not receive the 3\,dB gain of the $g$-function for the second component code, its result $\tilde{\cv}_1 =\operatorname{SC}_1(\ellv_{\nicefrac{N}{2}:N})$ is checked using the verification function 
\begin{equation}\label{eq:ryanscondition}
\tilde{\cv}_1 \stackrel{?}{=} \operatorname{hard}(\ellv_g)
\end{equation}
where $\ellv_g = g(\ellv_0, \ellv_1,\hat{\cv}_0)$ is the value of the $g$-function after the left branch decoded $\hat{\cv}_0$, and the hard decision function is given by
\begin{equation}
\operatorname{hard}(x) = 
\begin{cases}
0, & \text{if }  x \ge 0, \\
1, & \text{otherwise}.
\end{cases}
\end{equation}
Whenever the verification criterion is fulfilled, we know from the following lemma that the output of the speculative branch is exactly the result that conventional \ac{SC} decoding would have produced; therefore, the conventional decoding can be safely aborted.
Otherwise, the conventional decoding is continued until its result is available.
Hence, \ac{Spec-SC} decoding always outputs the same codeword estimate as conventional \ac{SC} decoding and, consequently, has identical error-correction performance \cite[Theorem 1]{seah2026specsc}.
\begin{lemma}\label{lemma:scio}
Let $\Cc$ be a polar code with length $N$ and $\operatorname{SC}(\cdot)$ its SC decoding function, and let $\ellv \in \mathbb{R}^N$ be an LLR vector. If $\xv:=\operatorname{hard}(\ellv) \in \Cc$, then $\operatorname{SC}(\ellv)= \xv$.
\end{lemma}
\begin{proof}
By induction. 

For the base case $N=1$, we have two cases: $\rv=\zerov$ or $\rv=\onev$.
If $\rv=\zerov$, $\operatorname{SC}(\ellv)=\boldsymbol{0}\:\forall \ellv$, so clearly also for $\ell_0\ge0$.
For $\rv=\onev$, $\operatorname{SC}(\ellv) = \operatorname{hard}(\ellv) $, i.e., exactly the decision function for the information bit.

Induction step: Let $\ellv=(\ellv_0\mid\ellv_1)$ and $\xv=(\xv_0\mid\xv_1)$. As $\xv$ is a polar codeword, we can write it as 
$\xv = (\uv \oplus \vv \mid \vv )$, with 
$\uv = (\xv_0 \oplus \xv_1) \in \Cc_0$ and 
$\vv=\xv_1\in \Cc_1$, where $\Cc_0$ and $\Cc_1$ are the component codes corresponding to the left and right branch, respectively. Let $\ellv_{f}= f(\ellv_0, \ellv_1)$. By the definition of the $f$-function, $\operatorname{hard}(\ellv_{f}) = \operatorname{hard}(\ellv_0)\oplus \operatorname{hard}(\ellv_1)=\xv_0 \oplus \xv_1=\uv\in \Cc_0$ and the conditions for the left branch are satisfied. Applying the induction hypothesis, we have $\operatorname{SC}_0(\ellv_{f})=\uv$.
Let $\ellv_{g}= g(\ellv_0, \ellv_1,\uv)$.
For the right branch, observe that as $\operatorname{hard}(\ellv_0)\oplus \operatorname{hard}(\ellv_1)=\uv$, both terms in the $g$-function have the same sign, and $\operatorname{hard}(\ellv_{g}) = \operatorname{hard}(\ellv_1) =\vv \in \Cc_1$ holds. Again using the induction hypothesis, we have $\operatorname{SC}_1(\ellv_{g})=\vv$. Finally, $\operatorname{SC}(\ellv)=(\uv \oplus \vv \mid \vv ) = \xv$.
\end{proof}

Continuing the normal \ac{SC} decoding procedure without waiting for the speculative right branch to complete is essential for maximizing speed-up, particularly for highly imbalanced trees. Furthermore, although the speculative right branch could in principle spawn further speculative decoding within itself, doing so typically yields poor success rates due to its 3\,dB worse operating point. Consequently, only the recursions along the main path (Task 1) invoke the speculative variants.

\section{Improved Acceptance Criteria}
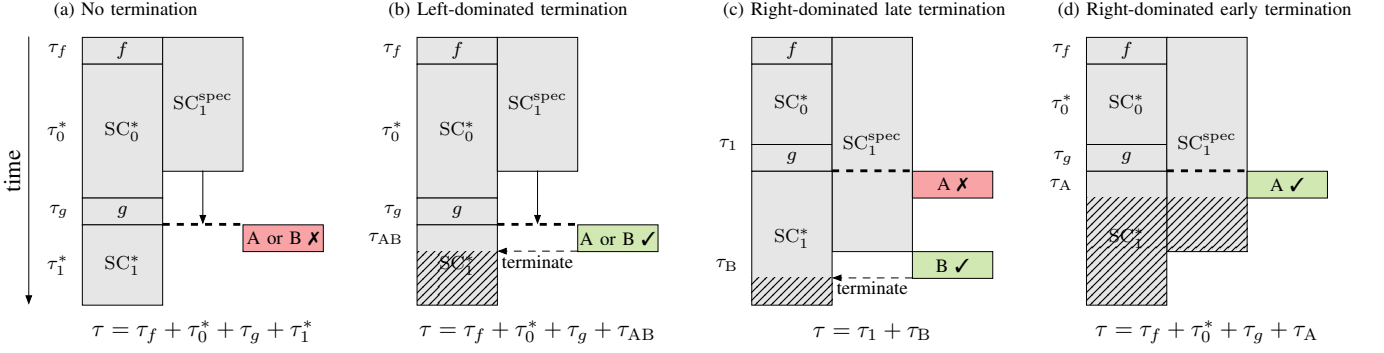
\begin{figure*}[t]
    \centering
    \resizebox{\linewidth}{!}{\begin{tikzpicture}[
 node distance = 8mm and 8mm]
	\tikzset{
		rect/.style={rectangle, draw, minimum width=1.2cm, fill=gray!20, font=\footnotesize, anchor=north},
        onecyc/.style={rectangle, draw, minimum width=1.2cm, minimum height=0.4cm, inner sep=0pt, fill=gray!20, font=\footnotesize,anchor=north},
        cond/.style={rectangle, draw, minimum width=1.2cm, minimum height=0.4cm, inner sep=0pt, font=\footnotesize,anchor=north},
        dspconn/.style={
        -{Triangle[length=1.4mm, width=1mm]}, %
        line cap=round,
        line join=round
      }
	}

    \draw[dspconn] (-2, 0) -> node[midway, above, rotate=90]{time} (-2, -4.0);

    \begin{scope}[xshift=0cm]
    \node[anchor=west] at (-1.75,.4) {\footnotesize (a) No termination};
    \node[onecyc] (f) at (-.6, 0) {$f$};
    \node[rect, minimum height=2.0cm] (l) at (-.6, -.4) {$\operatorname{SC}_0^*$};
    \node[onecyc] (g) at (-.6, -2.4) {$g$};
    \node[rect, minimum height=2.0cm]  at (.6, 0) {$\operatorname{SC}_1^\mathrm{spec}$};
    \draw[dspconn] (0.6, -2.0) -> (0.6, -2.8);
    \node[rect, minimum height=1.2cm] (r) at (-.6, -2.8) {$\operatorname{SC}_1^*$};
    \draw[very thick, dashed] (0, -2.8) -- (1.2, -2.8);   
    \node[cond, fill=ured!40] at (1.8, -2.8) {A or B \xmark};
    \node[left of=f,anchor=east,inner sep=0pt] {\footnotesize$\tau_f$};
    \node[left of=l,anchor=east,inner sep=0pt] {\footnotesize$\tau_0^*$};
    \node[left of=g,anchor=east,inner sep=0pt] {\footnotesize$\tau_g$};
    \node[left of=r,anchor=east,inner sep=0pt] {\footnotesize$\tau_1^*$};

    \node at (0.6, -4.4) {$\tau = \tau_f + \tau_0^* + \tau_g + \tau_1^*$};
    \end{scope}     

    \begin{scope}[xshift=5cm]
    \node[anchor=west] at (-1.75,.4) {\footnotesize (b) Left-dominated termination};
    \node[onecyc] (f) at (-.6, 0) {$f$};
    \node[rect, minimum height=2.0cm] (l) at (-.6, -.4) {$\operatorname{SC}_0^*$};
    \node[onecyc] (g) at (-.6, -2.4) {$g$};
    \node[rect, minimum height=2.0cm]  at (.6, 0) {$\operatorname{SC}_1^\mathrm{spec}$};
    \draw[dspconn] (0.6, -2.0) -> (0.6, -2.8);
    \node[rect, minimum height=1.2cm] at (-.6, -2.8) {$\operatorname{SC}_1^*$};
    \fill[pattern=north east lines] (-1.2,-3.2) rectangle (0,-4.0);
    \draw[very thick, dashed] (0, -2.8) -- (1.2, -2.8);   
    \node[cond, fill=green!40] (c) at (1.8, -2.8) {A or B \cmark};
    \draw[dspconn, dashed] (1.2, -3.2) -> node[below=-2pt]{\footnotesize terminate} (0, -3.2);
    \node[left of=f,anchor=east,inner sep=0pt] {\footnotesize$\tau_f$};
    \node[left of=l,anchor=east,inner sep=0pt] {\footnotesize$\tau_0^*$};
    \node[left of=g,anchor=east,inner sep=0pt] {\footnotesize$\tau_g$};
    \node[left of=c,xshift=-2.4cm,anchor=east, inner sep=0pt] {\footnotesize$\tau_\mathrm{AB}$};
    \node at (0.6, -4.4) {$\tau = \tau_f + \tau_0^* + \tau_g + \tau_\mathrm{AB}$};
    \end{scope}    

    \begin{scope}[xshift=10cm]
    \node[anchor=west] at (-1.75,.4) {\footnotesize (c) Right-dominated late termination};
    \node[onecyc] (f) at (-.6, 0) {$f$};
    \node[rect, minimum height=1.2cm] (l) at (-.6, -.4) {$\operatorname{SC}_0^*$};
    \node[onecyc] (g) at (-.6, -1.6) {$g$};
    \node[rect, minimum height=3.2cm] (r) at (.6, 0) {$\operatorname{SC}_1^\mathrm{spec}$};
    \node[rect, minimum height=2cm] at (-.6, -2) {$\operatorname{SC}_1^*$};
    \fill[pattern=north east lines] (-1.2,-3.6) rectangle (0,-4);
    \draw[very thick, dashed] (0, -2) -- (1.2, -2);
    \node[cond, fill=ured!40] at (1.8, -2) {A \xmark};
    \node[cond, fill=green!40] (c) at (1.8, -3.2) {B \cmark};
    \draw[dspconn, dashed] (1.2, -3.6) -> node[below=-2pt]{\footnotesize terminate} (0, -3.6);
    \node[left of=r,xshift=-1.2cm,anchor=east,inner sep=0pt] {\footnotesize$\tau_1$};
    \node[left of=c,xshift=-2.4cm,anchor=east,inner sep=0pt] {\footnotesize$\tau_\mathrm{B}$};
    \node at (0.6, -4.4) {$\tau = \tau_1 + \tau_\mathrm{B}$};
    
    \end{scope} 

    \begin{scope}[xshift=15cm]
    \node[anchor=west] at (-1.75,.4) {\footnotesize (d) Right-dominated early termination};
    \node[onecyc] (f) at (-.6, 0) {$f$};
    \node[rect, minimum height=1.2cm] (l) at (-.6, -.4) {$\operatorname{SC}_0^*$};
    \node[onecyc] (g) at (-.6, -1.6) {$g$};
    \node[rect, minimum height=3.2cm]  at (.6, 0) {$\operatorname{SC}_1^\mathrm{spec}$};
    \node[rect, minimum height=2cm] (r) at (-.6, -2) {$\operatorname{SC}_1^*$};
    \fill[pattern=north east lines] (-1.2,-2.4) rectangle (0,-4);
    \fill[pattern=north east lines] (0,-2.4) rectangle (1.2,-3.2);
    \draw[very thick, dashed] (0, -2) -- (1.2, -2);
    \node[cond, fill=green!40] (c) at (1.8, -2) {A \cmark};
    \node[left of=f,anchor=east,inner sep=0pt,anchor=east,inner sep=0pt] {\footnotesize$\tau_f$};
    \node[left of=l,anchor=east,inner sep=0pt] {\footnotesize$\tau_0^*$};
    \node[left of=g,anchor=east,inner sep=0pt] {\footnotesize$\tau_g$};
    \node[left of=c,xshift=-2.4cm,anchor=east,inner sep=0pt] {\footnotesize$\tau_\mathrm{A}$};
    \node at (0.6, -4.4) {$\tau = \tau_f + \tau_0^* + \tau_g + \tau_\mathrm{A}$};
    
    \end{scope}

\end{tikzpicture}}
    \caption{\footnotesize Timing of the concurrent decoding paths in the proposed \ac{Spec-SC} decoding of a single node for different scenarios. The speculative right branch, denoted by $\operatorname{SC}_1^\mathrm{spec}$, implements conventional \ac{FSSC}, while the recursive calls marked with an asterisk ($^*$) may include an internal speculation/early-termination mechanism. The hatched regions indicate computation terminated once at least one of the acceptance criteria triggered, reducing the node's decoding time.}
    \label{fig:timing}
\end{figure*}

The trigger probability of the original verification criterion is limited by two factors. First, for high-rate component codes, speculative execution of the right branch is more likely to fail because it is highly sensitive to channel noise in the absence of the additional 3\,dB gain provided by the left branch, whereas the output of the $g$-function generally constitutes a more reliable estimate. Second, for lower-rate component codes with stronger error-correction capability, \eqref{eq:ryanscondition} %
often fails to recognize that $\tilde{\cv}_1$ is correct because $\operatorname{hard}(\ellv_g) \neq \tilde{\cv}_1$. Indeed, \eqref{eq:ryanscondition} %
can trigger only when a subsequent decoding of $\ellv_g$ would not need to correct any errors. 

Fig.~\ref{fig:example1024} illustrates these limitations for the example of \ac{Spec-SC} decoding the root node of the 5G polar code with length $N=1024$. The original success probability $\operatorname{Pr}(\operatorname{hard}(\ellv_g)=\tilde{\cv}_1)$ is upper-bounded by both the correctness of the speculative result $\operatorname{Pr}(\tilde{\cv}_1=\cv_1)$ and of the $g$-function $\operatorname{Pr}(\operatorname{hard}(\ellv_g)=\cv_1)$, which dominate for the high-rate and medium rate codes, respectively. Due to the recursive structure of polar codes, nodes deeper in the tree behave similarly in qualitative terms, though at different operating points.

In the following, we introduce two additional acceptance criteria with higher trigger probabilities, enabling larger decoding speed-ups while preserving identical or comparable error-correcting performance. We refer to these as \emph{acceptance criteria} rather than \emph{verification criteria}, since they do not require the speculative branch to produce the same result as conventional \ac{SC} decoding.

\subsection{Code Membership-based Acceptance}
To tackle the first issue, we establish the following properties of \ac{Spec-SC} with the original acceptance criterion \eqref{eq:ryanscondition}.
\begin{lemma}\label{lemma:trigger}
For the original condition \eqref{eq:ryanscondition} to trigger, $\operatorname{hard}(\ellv_g)$ must be a codeword of the right component code $\Cc_1$, i.e.,
\begin{equation}
    \tilde{\cv}_1 = \operatorname{hard}(\ellv_g) \implies \operatorname{hard}(\ellv_g) \in \Cc_1.
\end{equation}
\end{lemma}
\begin{proof}
    The result follows immediately from $\operatorname{hard}(\ellv_g) = \tilde{\cv}_1 = \operatorname{SC}_1(\ellv_{\nicefrac{N}{2}:N})$ and \ac{SC} being a complete decoder.
\end{proof}

\begin{lemma}\label{lemma:same}
If $\tilde{\cv}_1 \ne \operatorname{hard}(\ellv_g)$ and $\operatorname{hard}(\ellv_g) \in \Cc_1$ then $\operatorname{SC}_1(\ellv_g) = \operatorname{hard}(\ellv_g)$.
\end{lemma}
\begin{proof}
   From Lemma~\ref{lemma:scio}, $\operatorname{SC}_1(\ellv_g) = \operatorname{hard}(\ellv_g)$ if $\operatorname{hard}(\ellv_g) \in \Cc_1$, regardless of $\tilde{\cv}_1 \ne \operatorname{hard}(\ellv_g)$.
\end{proof}

Based on these two properties, we define a new acceptance criterion solely based on checking code membership of $\operatorname{hard}(\ellv_g)$.
In particular, we accept  $\operatorname{hard}(\ellv_g)$ as the correct estimate for $\hat{\cv}_1$, if
\begin{equation} \label{eq:cwcondition}
\operatorname{hard}(\ellv_g) \stackrel{?}{\in} \Cc_1.
\end{equation}
In the following, we refer to this criterion as \emph{condition A}.
\begin{theorem}
The acceptance criterion \eqref{eq:cwcondition} triggers at least as frequently as \eqref{eq:ryanscondition}. Furthermore, the outputs of both decoder variants are identical. 
\end{theorem}

\begin{proof}
From Lemma~\ref{lemma:trigger}, $\tilde{\cv}_1 = \operatorname{hard}(\ellv_g) \implies \operatorname{hard}(\ellv_g) \in \Cc_1$.
Clearly, the output of the decoders could only be different in the new case $\tilde{\cv}_1 \ne \operatorname{hard}(\ellv_g)$ and $\operatorname{hard}(\ellv_g) \in \Cc_1$.
However, from Lemma~\ref{lemma:same}, $\hat{\cv}_1 = \operatorname{hard}(\ellv_g)= \operatorname{SC}_1(\ellv_g)$, which is identical to the result obtained by recomputing the second path.
\end{proof}

Note that, when employing only \eqref{eq:cwcondition},
the outcome $\tilde{\cv}_1$ of the speculative branch is not needed and the corresponding computation can be entirely omitted, reducing \ac{Spec-SC} to the per-node ``check-before-decoding'' strategy \cite{yoo2016checkconstituent, choi2017recusivesyndrome, kim2018earlystopping}.
Depending on the time required to evaluate the condition, it may still be beneficial to start decoding the right branch in parallel to the check. 

\subsection{Distance-based Acceptance}
The second issue, i.e., failing to detect correct decoding of the speculative branch, can be mitigated by weakening the strictness of \eqref{eq:ryanscondition}.
Especially when the right subtree has a modest amount of redundancy (many frozen bits), its error-correction capability ensures that even if $\operatorname{hard}(\ellv_g)$ is not exactly $\tilde{\cv}_1$, \ac{SC} decoding $\ellv_g$ likely results again in the speculative result $\tilde{\cv}_1$, so it should be accepted.
Therefore, we propose the following acceptance criteria based on the \ac{HD} and \ac{ED} \cite{valembois2004} between $\ellv_g$ and $\tilde{\cv}_1$
\begin{align}
\bigg(\sum\limits_{i}\tilde{c}_{1,i} \ne \operatorname{hard}(\ell_{g,i})\bigg) \stackrel{?}{\le} t_{\mathrm{H},\nu}\label{eq:hd}, \\
\bigg(\sum\limits_{i:\,\tilde{c}_{1,i} \ne \operatorname{hard}(\ell_{g,i})} |\ell_{g,i}| \bigg) \stackrel{?}{\le} t_{\mathrm{E},\nu},\label{eq:ed}
\end{align}
where $t_{\mathrm{H},\nu}$ and $t_{\mathrm{E},\nu}$ denote the acceptance radii of node $\nu$, which should be chosen according to the error-correcting capability of the subtree and the target operating point.
In general, increasing $t_{\mathrm{H},\nu}$ or $t_{\mathrm{E},\nu}$ is expected to make the decoder more optimistic, potentially yielding larger speed-ups, but also increasing the likelihood of accepting incorrect results and, consequently, degrading the error-correction performance. For $t_{\mathrm{H},\nu}=0$ (or $t_{\mathrm{E},\nu}=0$), the condition is equivalent to the original acceptance criterion \eqref{eq:ryanscondition}.
We will refer to the distance-based criteria as \emph{condition B}.

\begin{algorithm}[tbh]
	\caption{Speculative successive cancellation decoding with improved acceptance criteria.}\label{alg:specsci}
	\SetAlgoVlined
	\LinesNumbered
	\SetKwInOut{Input}{Input}\SetKwInOut{Output}{Output}
	\Input{Channel LLR vector $\ellv$, rate profile $\rv$, set of leaf decoders $\Dc$, acceptance thresholds $\tv_\mathrm{E}$.}
	\Output{Codeword estimate $\hat{\cv}$.}
	$\hat{\cv} = \operatorname{Spec-SC}(\ellv,\,\rv,\,0)$\;
	\algrule[.5pt]
	\Fn{$\operatorname{Spec-SC}(\ellv,\,\rv,\,\nu)$} {
		$N \gets \operatorname{length}(\ellv)$\;
		\eIf{leaf decoder for $\rv$ in $\Dc$}{
			\Return $\operatorname{decode}(\ellv, \rv)$ \;
		} {
			\Parallel{
				\Task{\normalfont 1}{
                    $\ellv_f = f(\ellv_{0:\nicefrac{N}{2}},\,\ellv_{\nicefrac{N}{2}:N})$\;
					$\hat{\cv}_0 \gets \operatorname{Spec-SC}(\ellv_f,\,\rv_{0:\nicefrac{N}{2}},\,2\nu+1)$\;
                    $\ellv_g \gets g(\ellv_{0:\nicefrac{N}{2}},\,\ellv_{\nicefrac{N}{2}:N},\, \hat{\cv}_0)$\;
                    \KwSet flag$_1$\;
					$\hat{\cv}_1 \gets \operatorname{Spec-SC}(\ellv_g,\,\rv_{\nicefrac{N}{2}:N},\,2\nu+2)$\;
                    \KwCancel Task 2, 3, 4\;
                    \Return $(\hat{\cv}_0 \oplus \hat{\cv}_1 \mid \hat{\cv}_1)$\;
				}
                \Task{\normalfont 2}{
					$\tilde{\cv}_1 \gets \operatorname{FSSC}( \ellv_{\nicefrac{N}{2}:N},\,\rv_{\nicefrac{N}{2}:N})$\;
                    \KwSet flag$_2$\;
				}
                \Task{\normalfont 3}{
                    \KwAwait ($\text{flag}_1$)\;
                    \If{$\operatorname{hard}(\ellv_g) \in \Cc(\rv_{\nicefrac{N}{2}:N}) $}{
            			\KwCancel Task 1, 2, 4\;
						\Return $(\hat{\cv}_0 \oplus \operatorname{hard}(\ellv_g) \mid \operatorname{hard}(\ellv_g))$\;
            		} 
				}
                \Task{\normalfont 4}{
                    \KwAwait ($\text{flag}_1 \wedge \text{flag}_2$)\;
                    \If{$\sum_{i:\,\tilde{c}_{1,i} \ne \operatorname{hard}(\ell_{g,i})} |\ell_{g,i}| \le t_{\mathrm{E},\nu}$}{ \label{algoline:crit}
                        \KwCancel Task 1, 2\;
						\Return $(\hat{\cv}_0 \oplus \tilde{\cv}_1 \mid \tilde{\cv}_1)$\;
					}
				}
			}
		}
	}
\end{algorithm}

Algorithm~\ref{alg:specsci} lists the pseudo-code for \ac{Spec-SC} with the two proposed acceptance criteria. The node indices $\nu$ are defined in breadth-first natural order, starting at the root node $\nu=0$. Note that $\tv_\mathrm{E}$ and line~\ref{algoline:crit} may be replaced by $\tv_\mathrm{H}$ and \eqref{eq:hd} for \ac{HD}-based detection, respectively.
Again, for minimal latency, the two criteria are implemented as parallel tasks.

\section{Latency Analysis and Threshold Optimization}
\subsection{Latency}
To compute the latency of the proposed decoding algorithm, we consider different cases based on the trigger probabilities of conditions A and B, and the latency of the constituent decoders, as shown in Fig.~\ref{fig:timing}.
At each node, let $\tau_0$ and $\tau_1$ denote the latency of conventional left and right \ac{FSSC} subdecoders, respectively, and $\tau_0^*$ and $\tau_1^*$ the latency of recursive decoders with a speculation or early termination feature.
Moreover, let $\tau_f$ and $\tau_g$ denote the latency of the $f$ and $g$ functions, respectively, and $\tau_A$ and $\tau_B$ that of the conditions $A$ and $B$, respectively.
For conventional \ac{FSSC} decoding, the latency of a node is computed as
\begin{equation}
    \tau_\mathrm{SC} = \tau_f + \tau_0 + \tau_g + \tau_1.
\end{equation}
The latency reduction of the proposed methods depends on the trigger probability of the acceptance condition at each node.
For code membership-based acceptance condition, the expected latency reduction can be computed as 
\begin{equation}
    \overline{\Delta\tau} = \tau_1 \cdot \operatorname{Pr}( \operatorname{hard}(\ellv_g) \in \Cc_1) - \tau_\mathrm{A,\Cc_1},
\end{equation}
where  $\tau_\mathrm{A,\Cc_1}$ is the latency overhead associated with code membership test for $\Cc_1$ itself.
Similarly, for nodes accepting the a speculative right branch decoding result, the latency is
\begin{equation}
    \tau_{\text{Spec-SC},\mathrm{B}} = \max(\tau_f + \tau_0 + \tau_g,\; \tau_1) + \tau_\mathrm{B},
\end{equation}
so the expected speedup can be computed as
\begin{equation}
    \overline{\Delta\tau} = (\tau_\mathrm{SC} - \tau_{\text{Spec-SC},\mathrm{B}}) \cdot \Pr(\mathrm{B}),
\end{equation}
where $\Pr(\mathrm{B})$ is the trigger probability of the used acceptance criterion.
In the following, we consider a simplified latency model in which only tree traversals are counted, i.e., $\tau_f = 1$. This is justified by fast \ac{SC} decoding architectures that precompute the $g$-function for both potential outcomes of $\hat{\cv}_0$ in parallel to the $f$-function, so that $\tau_g = 0$ \cite{zhang2012reducedlatency}. We further assume the checks to have negligible latency, i.e., $\tau_\mathrm{A}=\tau_\mathrm{B} = 0$, based on the availability of low-latency syndrome check algorithms \cite{choi2017recusivesyndrome}.
A more nuanced latency analysis would need to account for actual hardware implementations, which is beyond the scope of this paper. The numbers reported here are intended to indicate the maximal achievable gains of the proposed methods.

\subsection{Threshold Optimization}
Let $p_{\mathrm{e},\nu}(t)$ denote the probability of an error being introduced at node $\nu$ due to a wrong decoding result of the speculative right branch being accepted, i.e., condition B triggering falsely.
We want to optimize the thresholds such that the expected latency is reduced the most while limiting such newly introduced errors.
In general, the latency reduction and the error events are correlated, and therefore, joint optimization is challenging. To simplify the problem, we assume that speedups and new errors are independent. Since real, correlated error events can only reduce these newly introduced errors, this assumption guarantees that the actual \ac{BLER} degradation never exceeds the given threshold.
For a given \ac{SNR} operating point, we formulate the optimization problem as
\begin{equation}
    \begin{split}
    &\underset{t_0, \dots t_{N-2}}{\text{maximize}} \; \sum_\nu \overline{\Delta\tau}_\nu(t_\nu) \\
    &\text{subject to}\; \sum_\nu p_{\mathrm{e},\nu}(t_\nu) \le \epsilon,
    \end{split}
\end{equation}
where $\epsilon$ is the maximum allowed \ac{BLER} performance degradation.
Similarly to the optimization problem in \cite{pillet2026stopping}, this can be treated a discrete resource allocation problem given a dataset of error events and observed distances. This problem can be solved efficiently using dynamic programming \cite{Bellman1957}.

\section{Numerical Results}
We consider $(4096,K)$ polar codes with code rates $R\in\{\nicefrac{1}{2},\nicefrac{3}{4},\nicefrac{7}{8}\}$ designed using Gaussian approximation~\cite{Trifonov2012Efficient} at design \ac{SNR} of $3\,\text{dB}$, over the \ac{BI-AWGN} channel.
The distance-based acceptance criterion is applied only at nodes with positive potential speed-up $\tau_1>\tau_{\mathrm{B}}$ and the acceptance thresholds have been optimized at each \ac{SNR} point to not increase the error rate by more than $50\,\%$, i.e., $\epsilon = 0.5 \cdot \text{BLER}$. All decoders use the fast nodes $\Dc = \{\text{Rate-0}, \text{REP}, \text{SPC}, \text{Rate-1}\}$.
\subsection{Error-rate Performance}
\begin{figure}[t]
    \centering
    \begin{tikzpicture}
\begin{axis}[
    width=\linewidth,
    height=.75\linewidth,
    legend columns=1,
    legend style={
        font=\footnotesize,
        fill opacity=0.75,
        draw opacity=1,
        text opacity=1,
    },
    legend cell align={left},
    ylabel={Block error rate (BLER)},
    xlabel={$E_\mathrm{b}/N_0$ in dB},
    xmajorgrids,
	yminorticks=true,
	ymajorgrids,
    grid style={dotted,dgray},
    ymin=1e-4,
	ymax=1.05,
    ymode=log,
    xmax=5.5,
    xmin=1
]

\addplot [black, line width=1pt, mark=none] 
table[col sep=comma]{%
0.50, 1.000e+00
0.75, 9.962e-01
1.00, 9.560e-01
1.25, 8.492e-01
1.50, 5.531e-01
1.75, 2.437e-01
2.00, 7.699e-02
2.25, 1.813e-02
2.50, 2.295e-03
2.75, 2.312e-04
3.00, 1.260e-05
};
\addlegendentry{SC / Spec-SC \cite{seah2026specsc} / A \cite{yoo2016checkconstituent}}

\addplot [purple, line width=1pt, mark=square] 
table[col sep=comma]{%
1.00, 1.000e+00
1.25, 1.000e+00
1.50, 8.119e-01
1.75, 3.442e-01
2.00, 1.112e-01
2.25, 2.076e-02
2.50, 2.837e-03
2.75, 3.202e-04
3.00, 1.552e-05
};
\addlegendentry{Spec-SC A+B (HD)}

\addplot [mblue, line width=1pt, mark=o] 
table[col sep=comma]{%
1.00, 1.000e+00
1.25, 1.000e+00
1.50, 8.248e-01
1.75, 3.679e-01
2.00, 1.118e-01
2.25, 2.194e-02
2.50, 2.791e-03
2.75, 3.193e-04
3.00, 1.709e-05
};
\addlegendentry{Spec-SC A+B (ED)}

\addplot [black, line width=1pt, mark=none] 
table[col sep=comma]{%
2.00, 9.961e-01
2.50, 6.537e-01
3.00, 7.707e-02
3.50, 2.350e-03
4.00, 5.815e-05
};

\addplot [purple, line width=1pt, mark=square] 
table[col sep=comma]{%
2.00, 1.000e+00
2.50, 8.938e-01
3.00, 8.714e-02
3.50, 2.840e-03
4.00, 6.077e-05
};

\addplot [mblue, line width=1pt, mark=o] 
table[col sep=comma]{%
2.00, 1.000e+00
2.50, 9.353e-01
3.00, 9.573e-02
3.50, 2.823e-03
4.00, 8.295e-05
};

\addplot [black, line width=1pt, mark=none] 
table[col sep=comma]{%
3.00, 9.924e-01
3.50, 5.614e-01
4.00, 7.576e-02
4.50, 5.029e-03
5.00, 3.252e-04
5.50, 2.305e-05
};

\addplot [purple, line width=1pt, mark=square] 
table[col sep=comma]{%
3.00, 1.000e+00
3.50, 8.106e-01
4.00, 9.701e-02
4.50, 6.118e-03
5.00, 3.811e-04
5.50, 5.151e-05
};

\addplot [mblue, line width=1pt, mark=o] 
table[col sep=comma]{%
3.00, 1.000e+00
3.50, 7.957e-01
4.00, 1.166e-01
4.50, 7.328e-03
5.00, 4.387e-04
5.50, 4.684e-05
};

\node[rotate=-75] at (axis cs:2.35, 1e-3) {\footnotesize $(4096,2048)$};
\node[rotate=-70] at (axis cs:3.4, 1e-3) {\footnotesize $(4096,3072)$};
\node[rotate=-65] at (axis cs:4.6, 1e-3) {\footnotesize $(4096,3584)$};

\end{axis}

\end{tikzpicture}
    \caption{\footnotesize \Ac{BLER} performance of \ac{Spec-SC} with \ac{HD}- and \ac{ED}-based acceptance conditions vs. conventional \ac{SC} decoding.}
    \label{fig:bler}
\end{figure}
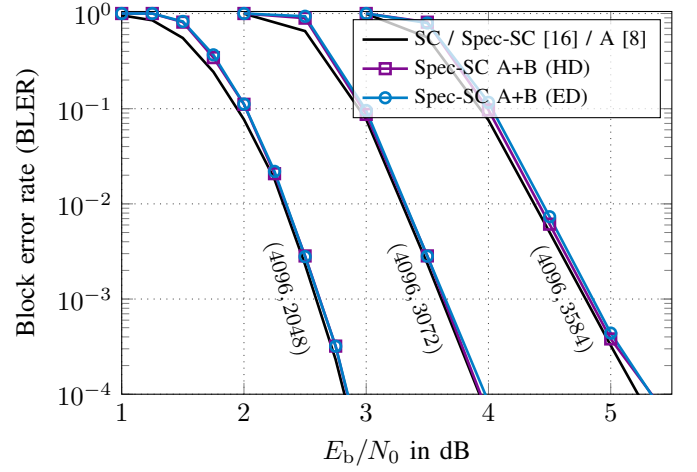
Fig.~\ref{fig:bler} compares the \ac{BLER} performance of \ac{Spec-SC} with the proposed acceptance criteria to conventional \ac{SC} decoding. The conventional \ac{SC} curve coincides with \ac{Spec-SC} using the original condition or only using condition A.
As expected, the curve of the proposed algorithm follows conventional decoding closely, well within the allowed performance degradation of $50\,\%$.

\subsection{Trigger Probabilities}
\begin{figure*}[!htb]
    \centering
    \resizebox{\linewidth}{!}{\begin{tikzpicture}[
  every node/.style={draw, rectangle, align=center, font=\footnotesize, minimum width=2.1cm, inner sep=2pt},
  dspconn/.style={
        -{Triangle[length=1.4mm, width=1mm]}, %
        line cap=round,
        line join=round
      }
]
  \node (n0) at (14.37,0.00) {$(4096,3584)$ \\ $\Pr(\mathrm{A})  = \mathbf{0.17}$ \\ $\Pr(\mathrm{B})  = 0.00$ \\ $\Pr(\mathrm{AB}) = 0.00$ \\ $t_{\mathrm{E},0} = 0.00$};
  \node (n1) at (7.76,-2.10) {$(2048,1586)$ \\ $\Pr(\mathrm{A})  = \mathbf{0.04}$ \\ $\Pr(\mathrm{B})  = 0.00$ \\ $\Pr(\mathrm{AB}) = 0.00$ \\ $t_{\mathrm{E},1} = 22.47$};
  \node (n2) at (20.99,-2.10) {$(2048,1998)$ \\ $\Pr(\mathrm{A})  = \mathbf{0.99}$ \\ $\Pr(\mathrm{B})  = 0.58$ \\ $\Pr(\mathrm{AB}) = 0.58$ \\ $t_{\mathrm{E},2} = 0.00$};
  \node (n3) at (3.45,-4.20) {$(1024,622)$ \\ $\Pr(\mathrm{A})  = 0.01$ \\ $\Pr(\mathrm{B})  = \mathbf{0.05}$ \\ $\Pr(\mathrm{AB}) = 0.00$ \\ $t_{\mathrm{E},3} = 16.18$};
  \node (n4) at (12.07,-4.20) {$(1024,964)$ \\ $\Pr(\mathrm{A})  = \mathbf{0.97}$ \\ $\Pr(\mathrm{B})  = 0.48$ \\ $\Pr(\mathrm{AB}) = 0.47$ \\ $t_{\mathrm{E},4} = 4.27$};
  \node (n5) at (18.97,-4.20) {$(1024,975)$ \\ $\Pr(\mathrm{A})  = \mathbf{0.99}$ \\ $\Pr(\mathrm{B})  = 0.57$ \\ $\Pr(\mathrm{AB}) = 0.57$ \\ $t_{\mathrm{E},5} = 0.00$};
  \node (n6) at (23.00,-4.20) {$(1024,1023)$};
  \node (n7) at (1.15,-6.30) {$(512,195)$ \\ $\Pr(\mathrm{A})  = 0.00$ \\ $\Pr(\mathrm{B})  = \mathbf{0.47}$ \\ $\Pr(\mathrm{AB}) = 0.00$ \\ $t_{\mathrm{E},7} = 23.31$};
  \node (n8) at (5.75,-6.30) {$(512,427)$ \\ $\Pr(\mathrm{A})  = \mathbf{0.86}$ \\ $\Pr(\mathrm{B})  = 0.55$ \\ $\Pr(\mathrm{AB}) = 0.50$ \\ $t_{\mathrm{E},8} = 3.78$};
  \node (n9) at (10.35,-6.30) {$(512,455)$ \\ $\Pr(\mathrm{A})  = \mathbf{0.95}$ \\ $\Pr(\mathrm{B})  = 0.62$ \\ $\Pr(\mathrm{AB}) = 0.61$ \\ $t_{\mathrm{E},9} = 2.14$};
  \node (n10) at (13.80,-6.30) {$(512,509)$};
  \node (n11) at (17.25,-6.30) {$(512,464)$ \\ $\Pr(\mathrm{A})  = \mathbf{0.98}$ \\ $\Pr(\mathrm{B})  = 0.74$ \\ $\Pr(\mathrm{AB}) = 0.73$ \\ $t_{\mathrm{E},11} = 2.50$};
  \node (n12) at (20.70,-6.30) {$(512,511)$};
  \node (n15) at (0.00,-8.40) {$(256,39)$ \\ $\Pr(\mathrm{A})  = 0.00$ \\ $\Pr(\mathrm{B})  = \mathbf{0.94}$ \\ $\Pr(\mathrm{AB}) = 0.00$ \\ $t_{\mathrm{E},15} = 33.96$};
  \node (n16) at (2.30,-8.40) {$(256,156)$ \\ $\Pr(\mathrm{A})  = 0.55$ \\ $\Pr(\mathrm{B})  = \mathbf{0.71}$ \\ $\Pr(\mathrm{AB}) = 0.43$ \\ $t_{\mathrm{E},16} = 8.89$};
  \node (n17) at (4.60,-8.40) {$(256,181)$ \\ $\Pr(\mathrm{A})  = \mathbf{0.76}$ \\ $\Pr(\mathrm{B})  = 0.68$ \\ $\Pr(\mathrm{AB}) = 0.56$ \\ $t_{\mathrm{E},17} = 6.71$};
  \node (n18) at (6.90,-8.40) {$(256,246)$};
  \node (n19) at (9.20,-8.40) {$(256,204)$ \\ $\Pr(\mathrm{A})  = \mathbf{0.90}$ \\ $\Pr(\mathrm{B})  = 0.73$ \\ $\Pr(\mathrm{AB}) = 0.68$ \\ $t_{\mathrm{E},19} = 3.42$};
  \node (n20) at (11.50,-8.40) {$(256,251)$};
  \node (n23) at (16.10,-8.40) {$(256,211)$ \\ $\Pr(\mathrm{A})  = \mathbf{0.97}$ \\ $\Pr(\mathrm{B})  = 0.90$ \\ $\Pr(\mathrm{AB}) = 0.88$ \\ $t_{\mathrm{E},23} = 3.95$};
  \node (n24) at (18.40,-8.40) {$(256,253)$};

  \draw[dspconn](n0) -- (n1);
  \draw[dspconn](n0) -- (n2);
  \draw[dspconn](n1) -- (n3);
  \draw[dspconn](n1) -- (n4);
  \draw[dspconn](n2) -- (n5);
  \draw[dspconn](n2) -- (n6);
  \draw[dspconn](n3) -- (n7);
  \draw[dspconn](n3) -- (n8);
  \draw[dspconn](n4) -- (n9);
  \draw[dspconn](n4) -- (n10);
  \draw[dspconn](n5) -- (n11);
  \draw[dspconn](n5) -- (n12);
  \draw[dspconn](n7) -- (n15);
  \draw[dspconn](n7) -- (n16);
  \draw[dspconn](n8) -- (n17);
  \draw[dspconn](n8) -- (n18);
  \draw[dspconn](n9) -- (n19);
  \draw[dspconn](n9) -- (n20);
  \draw[dspconn](n11) -- (n23);
  \draw[dspconn](n11) -- (n24);
\end{tikzpicture}}
    \caption{\footnotesize Empirical trigger probabilities of conditions A and B (\ac{ED}) in \ac{Spec-SC} decoding the $(4096,3584)$ polar code at $E_\mathrm{b}/N_0 = 4.5\,\text{dB}$.}
    \label{fig:trigger}
\end{figure*}
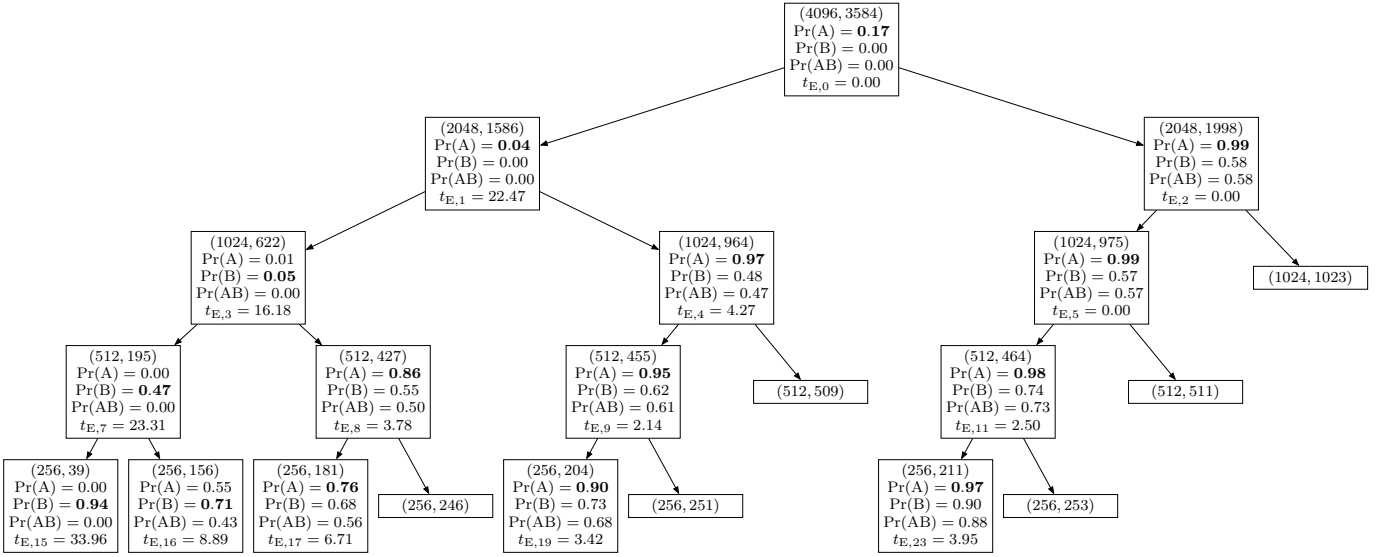
\begin{figure*}[!htb]
    \centering
        \resizebox{\linewidth}{!}{\input{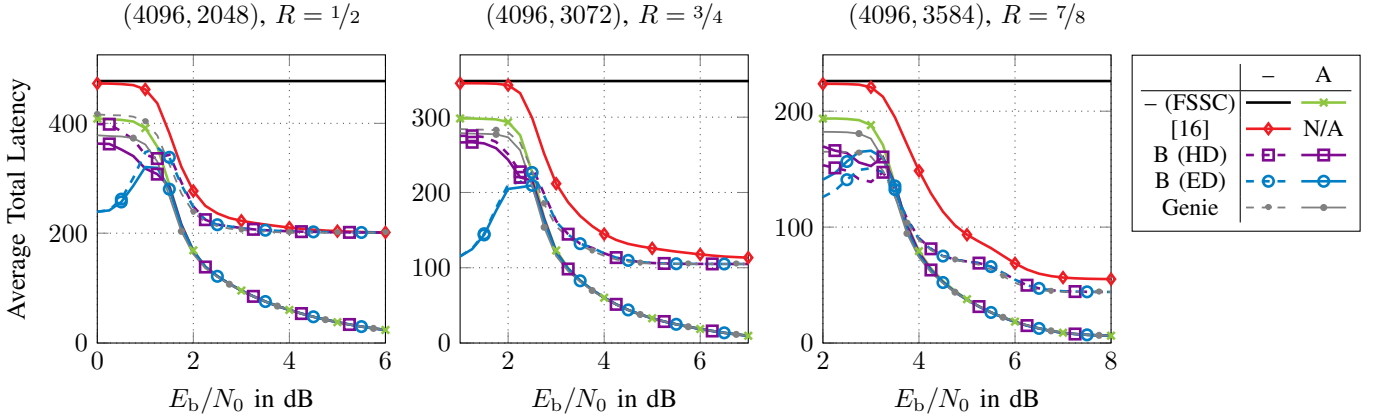}}
    \caption{\footnotesize Average decoding latency (serially dependent visited nodes) of different \ac{Spec-SC} configurations vs. conventional \ac{FSSC} decoding.}
    \label{fig:latency}
\end{figure*}

Fig.~\ref{fig:trigger} shows a partial decoding tree of the $(4096,3584)$ polar code, with the trigger probabilities of conditions A and B at each node, evaluated at $E_\mathrm{b}/N_0 = 4.5\,\text{dB}$ for the given \ac{ED} thresholds $t_{\mathrm{E},\nu}$. $\Pr(\mathrm{AB})$ denotes the probability of both conditions being satisfied simultaneously, which is equivalent to the trigger probability of the original acceptance criterion \eqref{eq:ryanscondition}. Code membership proves to be an effective early termination criterion for the high-rate nodes, whereas the lower-rate nodes benefit more from the distance-based acceptance condition.

\subsection{Decoding Latency}
Fig.~\ref{fig:latency} shows the average total latency in terms of serially dependent node traversals for the different decoding configurations with respect to \ac{SNR}.
As a baseline, we also consider a genie-based acceptance criterion ($\hat{\cv}_1=\tilde{\cv}_1$) that marks the maximal speed-up without error-rate performance degradation.
The qualitative behavior is the same across all code rates, while the achievable gains are generally higher for higher code rates.
Both distance-based criteria closely track the genie acceptance, indicating that the proposed scheme is essentially optimal, with the soft \ac{ED} outperforming the simpler \ac{HD}-based condition at low \ac{SNR}.
Without code membership checking, latency plateaus at the combined latency of all normal \ac{FSSC} right branches (right-dominated late termination).
At \acp{SNR} corresponding to \ac{BLER} $=10^{-3}$, it achieves throughput gains of $2.2\times$, $2.7\times$, and $3.2\times$ over \ac{FSSC} and $1.1\times$, $1.2\times$, and $1.4\times$ over original \ac{Spec-SC} for the rates $\nicefrac{1}{2}$, $\nicefrac{3}{4}$, and $\nicefrac{7}{8}$, respectively.
The code membership check (condition A) produces a further $1.9\times$, $1.7\times$, and $1.7\times$ throughput increase, respectively.
\begin{figure}[htb]
    \centering
    \begin{tikzpicture}
\begin{axis}[
    width=\linewidth,
    height=.8\linewidth,
    legend columns=1,
    legend style={
        fill opacity=0.75,
        draw opacity=1,
        text opacity=1,
        inner sep=2pt,
        /tikz/every node/.append style={scale=0.75}
    },
    legend cell align={left},
    legend pos={south east},
    xlabel={Latency $\tau$ (serially dependent visited nodes)},
    ylabel={Fraction of Samples $\leq \tau$},
    xmajorgrids,
	yminorticks=true,
	ymajorgrids,
    grid style={dotted,dgray},
    ymin=0,
	ymax=1,
    xmax=240,
    xmin=40
]

\addplot[black, thick, mark=none] coordinates {(226,0)(226,1)};
\addlegendentry{FSSC}
\addplot[ured, thick, mark=diamond, mark repeat=8, const plot mark right] coordinates {(85,0)(85,0.001)(91,0.0015)(92,0.002)(93,0.0025)(95,0.003)(98,0.0045)(99,0.005)(100,0.0055)(101,0.0065)(102,0.008)(103,0.009)(104,0.0105)(105,0.012)(106,0.0145)(107,0.017)(108,0.0195)(109,0.022)(110,0.0275)(111,0.0315)(112,0.0355)(113,0.039)(114,0.043)(115,0.05)(116,0.0575)(117,0.0645)(118,0.0705)(119,0.0825)(120,0.087)(121,0.0925)(122,0.1035)(123,0.1135)(124,0.125)(125,0.137)(126,0.147)(127,0.159)(128,0.1735)(129,0.1855)(130,0.2005)(131,0.214)(132,0.228)(133,0.247)(134,0.2615)(135,0.2765)(136,0.295)(137,0.3165)(138,0.33)(139,0.3455)(140,0.371)(141,0.3845)(142,0.406)(143,0.4265)(144,0.4475)(145,0.4665)(146,0.4835)(147,0.4985)(148,0.521)(149,0.5415)(150,0.559)(151,0.5735)(152,0.591)(153,0.611)(154,0.629)(155,0.6435)(156,0.6615)(157,0.6835)(158,0.703)(159,0.7225)(160,0.732)(161,0.746)(162,0.766)(163,0.7815)(164,0.7935)(165,0.8075)(166,0.819)(167,0.8335)(168,0.8415)(169,0.854)(170,0.864)(171,0.874)(172,0.8825)(173,0.892)(174,0.901)(175,0.913)(176,0.919)(177,0.9275)(178,0.9295)(179,0.9365)(180,0.9405)(181,0.946)(182,0.952)(183,0.9545)(184,0.9605)(185,0.9665)(186,0.9695)(187,0.9725)(188,0.976)(189,0.978)(190,0.98)(191,0.981)(192,0.982)(193,0.983)(194,0.985)(195,0.9855)(196,0.9865)(197,0.9875)(198,0.9885)(199,0.989)(200,0.9905)(202,0.991)(203,0.9915)(204,0.9935)(205,0.994)(206,0.9945)(208,0.9955)(209,0.9965)(210,0.997)(211,0.9975)(212,0.9985)(215,0.999)(216,0.9995)(218,1)};
\addlegendentry{Spec-SC \cite{seah2026specsc}}

\addplot[mblue, thick, dashed, mark=o, mark options={solid}, mark repeat=8, const plot mark right] coordinates {(69,0)(69,0.025)(70,0.029)(71,0.0745)(72,0.1065)(73,0.141)(74,0.177)(75,0.2095)(76,0.221)(77,0.2325)(78,0.24)(79,0.259)(80,0.2655)(81,0.278)(82,0.323)(83,0.3385)(84,0.3865)(85,0.4305)(86,0.4735)(87,0.512)(88,0.547)(89,0.571)(90,0.5925)(91,0.61)(92,0.643)(93,0.661)(94,0.6855)(95,0.7225)(96,0.745)(97,0.7645)(98,0.7885)(99,0.809)(100,0.821)(101,0.831)(102,0.8485)(103,0.8545)(104,0.8655)(105,0.878)(106,0.889)(107,0.9005)(108,0.912)(109,0.919)(110,0.9245)(111,0.93)(112,0.9335)(113,0.9345)(114,0.9375)(115,0.94)(116,0.9445)(117,0.951)(118,0.9545)(119,0.958)(120,0.9615)(121,0.9645)(122,0.966)(123,0.9675)(124,0.969)(125,0.97)(126,0.9715)(127,0.9735)(128,0.975)(129,0.9765)(130,0.9785)(131,0.9795)(132,0.9805)(134,0.9815)(135,0.984)(137,0.985)(139,0.9855)(140,0.986)(141,0.987)(142,0.9875)(143,0.989)(144,0.9895)(145,0.99)(146,0.9915)(151,0.992)(152,0.993)(153,0.994)(156,0.9945)(158,0.995)(159,0.996)(163,0.9965)(164,0.9975)(165,0.9985)(166,0.999)(167,0.9995)(178,1)};
\addlegendentry{Spec-SC B (ED)}

\addplot[mgray, thick, dashed, mark=*, mark size=1pt, mark options={solid}, mark repeat=8, const plot mark right] coordinates {(69,0)(69,0.0315)(70,0.036)(71,0.085)(72,0.121)(73,0.1595)(74,0.1965)(75,0.2295)(76,0.2395)(77,0.248)(78,0.2525)(79,0.274)(80,0.2785)(81,0.2935)(82,0.342)(83,0.356)(84,0.4135)(85,0.4625)(86,0.509)(87,0.5525)(88,0.589)(89,0.614)(90,0.634)(91,0.643)(92,0.6775)(93,0.6915)(94,0.717)(95,0.76)(96,0.7775)(97,0.8075)(98,0.833)(99,0.851)(100,0.8595)(101,0.8655)(102,0.879)(103,0.8855)(104,0.8975)(105,0.908)(106,0.9155)(107,0.9255)(108,0.9335)(109,0.9395)(110,0.942)(111,0.945)(112,0.9475)(113,0.951)(114,0.9565)(115,0.9605)(116,0.9655)(117,0.97)(118,0.9705)(119,0.9725)(120,0.973)(121,0.974)(122,0.9765)(124,0.979)(125,0.9805)(127,0.9825)(128,0.9845)(131,0.9855)(132,0.986)(133,0.9875)(134,0.9885)(136,0.989)(137,0.9895)(138,0.99)(139,0.9905)(140,0.991)(141,0.9925)(142,0.9935)(145,0.994)(147,0.9945)(151,0.9955)(152,0.996)(153,0.9965)(154,0.997)(156,0.9975)(157,0.998)(160,0.999)(161,0.9995)(162,1)};
\addlegendentry{Spec-SC Genie}

\addplot[green, thick, mark=x, mark repeat=8, const plot mark right] coordinates {(41,0)(41,0.0005)(42,0.001)(43,0.0015)(45,0.0025)(46,0.004)(47,0.006)(48,0.0065)(49,0.0085)(50,0.012)(51,0.0145)(52,0.019)(53,0.023)(54,0.0275)(55,0.0355)(56,0.0435)(57,0.0585)(58,0.07)(59,0.0875)(60,0.104)(61,0.1135)(62,0.1285)(63,0.15)(64,0.1705)(65,0.1955)(66,0.225)(67,0.2475)(68,0.269)(69,0.2975)(70,0.3245)(71,0.3525)(72,0.387)(73,0.417)(74,0.4515)(75,0.481)(76,0.5145)(77,0.543)(78,0.5685)(79,0.6015)(80,0.627)(81,0.6535)(82,0.6795)(83,0.7045)(84,0.7245)(85,0.745)(86,0.7665)(87,0.785)(88,0.7995)(89,0.82)(90,0.833)(91,0.8465)(92,0.8595)(93,0.8695)(94,0.88)(95,0.889)(96,0.8955)(97,0.9025)(98,0.908)(99,0.9145)(100,0.919)(101,0.9235)(102,0.927)(103,0.9295)(104,0.933)(105,0.936)(106,0.941)(107,0.9425)(108,0.944)(109,0.945)(110,0.948)(111,0.95)(112,0.951)(113,0.954)(114,0.9555)(115,0.956)(116,0.9575)(117,0.9585)(118,0.9605)(119,0.9615)(120,0.963)(121,0.964)(122,0.9655)(123,0.966)(124,0.967)(127,0.969)(128,0.971)(129,0.9725)(130,0.9735)(132,0.9745)(133,0.9755)(134,0.977)(135,0.9775)(136,0.978)(137,0.979)(138,0.981)(139,0.9825)(140,0.983)(141,0.984)(142,0.985)(143,0.986)(145,0.987)(148,0.9875)(149,0.988)(151,0.9885)(152,0.99)(153,0.9905)(154,0.992)(156,0.9925)(157,0.993)(161,0.9935)(163,0.994)(165,0.9955)(166,0.9965)(169,0.997)(170,0.9975)(172,0.998)(175,0.9985)(177,0.999)(178,1)};
\addlegendentry{Spec-SC A}

\addplot[mblue, thick, solid, mark=o, mark options={solid}, mark repeat=8, const plot mark right] coordinates {(39,0)(39,0.0005)(42,0.001)(43,0.0015)(44,0.002)(45,0.004)(46,0.007)(47,0.0085)(48,0.009)(49,0.0115)(50,0.016)(51,0.0195)(52,0.0255)(53,0.0305)(54,0.039)(55,0.0515)(56,0.0655)(57,0.0875)(58,0.1065)(59,0.128)(60,0.1525)(61,0.1675)(62,0.187)(63,0.2115)(64,0.2395)(65,0.2675)(66,0.299)(67,0.33)(68,0.368)(69,0.404)(70,0.4345)(71,0.4725)(72,0.511)(73,0.542)(74,0.569)(75,0.6005)(76,0.6345)(77,0.6595)(78,0.685)(79,0.7125)(80,0.7355)(81,0.7555)(82,0.775)(83,0.7915)(84,0.8105)(85,0.8275)(86,0.841)(87,0.8495)(88,0.862)(89,0.8725)(90,0.8825)(91,0.889)(92,0.8965)(93,0.9015)(94,0.9085)(95,0.9145)(96,0.9175)(97,0.9215)(98,0.924)(99,0.9275)(100,0.929)(101,0.933)(102,0.9345)(103,0.937)(104,0.942)(105,0.943)(106,0.9445)(107,0.946)(108,0.948)(109,0.949)(110,0.95)(111,0.9505)(112,0.952)(113,0.953)(114,0.955)(115,0.9555)(116,0.9565)(117,0.9585)(118,0.9595)(119,0.96)(120,0.9615)(121,0.9625)(122,0.9635)(123,0.965)(124,0.966)(126,0.967)(127,0.9685)(128,0.97)(129,0.971)(130,0.9735)(131,0.9745)(133,0.975)(134,0.977)(135,0.978)(136,0.9785)(137,0.979)(138,0.981)(139,0.9815)(140,0.9825)(141,0.9845)(142,0.985)(143,0.9855)(144,0.9865)(145,0.988)(146,0.9885)(147,0.9895)(148,0.9905)(149,0.991)(151,0.9915)(152,0.992)(154,0.9925)(157,0.9935)(159,0.994)(160,0.9945)(161,0.995)(163,0.996)(164,0.9965)(168,0.997)(169,0.9975)(170,0.998)(171,0.9985)(174,0.999)(176,0.9995)(188,1)};
\addlegendentry{Spec-SC A+B (ED)}

\addplot[mgray, thick, solid, mark=*, mark size=1pt, mark options={solid}, mark repeat=8, const plot mark right] coordinates {(38,0)(38,0.0005)(39,0.001)(41,0.0015)(43,0.002)(44,0.003)(45,0.005)(46,0.0085)(47,0.0095)(48,0.01)(49,0.013)(50,0.018)(51,0.0215)(52,0.028)(53,0.034)(54,0.0475)(55,0.058)(56,0.073)(57,0.0965)(58,0.1175)(59,0.1415)(60,0.167)(61,0.186)(62,0.2095)(63,0.2375)(64,0.27)(65,0.303)(66,0.3395)(67,0.3785)(68,0.42)(69,0.458)(70,0.4955)(71,0.533)(72,0.5635)(73,0.598)(74,0.628)(75,0.6565)(76,0.683)(77,0.709)(78,0.737)(79,0.76)(80,0.781)(81,0.8035)(82,0.8185)(83,0.8335)(84,0.848)(85,0.859)(86,0.8725)(87,0.882)(88,0.892)(89,0.9005)(90,0.907)(91,0.914)(92,0.92)(93,0.9275)(94,0.9295)(95,0.9315)(96,0.934)(97,0.939)(98,0.942)(99,0.9435)(100,0.944)(101,0.9455)(102,0.9495)(103,0.951)(104,0.9545)(105,0.955)(106,0.957)(107,0.9585)(108,0.96)(110,0.9605)(111,0.961)(112,0.9615)(113,0.962)(114,0.9645)(115,0.965)(116,0.9665)(117,0.9675)(118,0.9685)(119,0.9705)(120,0.971)(121,0.9715)(123,0.973)(124,0.975)(125,0.9755)(126,0.9765)(127,0.9785)(129,0.9805)(130,0.981)(131,0.9825)(132,0.9835)(133,0.985)(135,0.9855)(136,0.986)(137,0.988)(138,0.9885)(139,0.989)(140,0.99)(141,0.991)(143,0.9915)(144,0.9925)(145,0.993)(148,0.9935)(149,0.9945)(150,0.995)(151,0.9955)(154,0.996)(156,0.9965)(157,0.997)(158,0.9975)(162,0.998)(163,0.9985)(164,0.999)(165,0.9995)(167,1)};
\addlegendentry{Spec-SC A+Genie}

\coordinate (m-sc)   at (axis cs:226,0.5);
\coordinate (m-spec) at (axis cs:147.1,0.5);
\coordinate (m-b)    at (axis cs:86.7,0.5);
\coordinate (m-ab)   at (axis cs:71.7,0.5);

\end{axis}

\draw[->, thick, dashed] (m-sc) -- (m-spec)
    node[midway, above=3pt, fill=white, fill opacity=0.8, text opacity=1] {34.9\,\%};
\draw[->, thick, dashed] (m-spec) -- (m-b)
    node[midway, above=3pt, fill=white, fill opacity=0.8, text opacity=1] {41.1\,\%};
\draw[->, thick, dashed] (m-b) -- (m-ab)
    node[midway, above=3pt, fill=white, fill opacity=0.8, text opacity=1] {17.3\,\%};
\end{tikzpicture}
    \caption{\footnotesize Empirical cumulative distribution function of the latency of different \ac{Spec-SC} configurations vs. \ac{FSSC} decoding for the $(4096,3584)$ polar code at $E_\mathrm{b}/N_0 = 4\,\text{dB}$.}
    \label{fig:latency_cdf}
\end{figure}
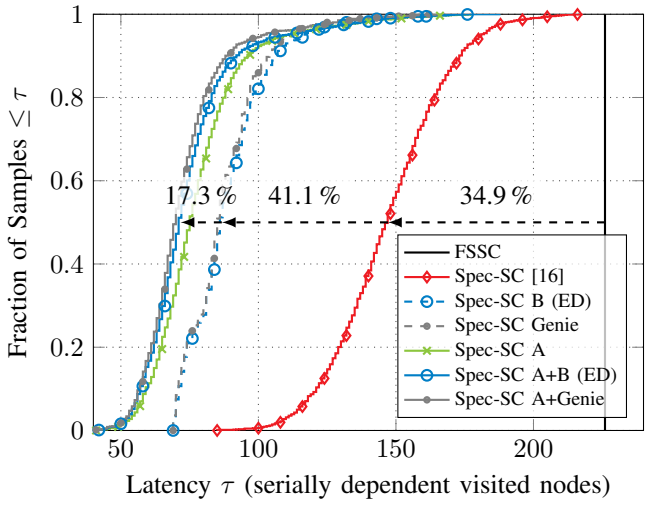
Fig.~\ref{fig:latency_cdf} shows a more detailed latency analysis for the $(4096,3584)$ code at $E_\mathrm{b}/N_0 = 4\,\text{dB}$ in the form of an empirical cumulative distribution function.
The gains from the respective methods (condition B, condition A+B) are relative reductions.
While the majority of the overall latency reduction is due to the code membership check (condition A), improving beyond this performance requires speculation.

\section{Conclusion and Outlook}
In this paper, we showed that the original acceptance criterion of \ac{Spec-SC} can be reformulated in terms of code membership, which results in an acceptance criterion with a strictly higher trigger probability. Moreover, we propose distance-based acceptance criteria that can be used to further speed-up \ac{Spec-SC} decoding, where thresholds can be adjusted to balance speed-up and \ac{BLER} performance.
Simulation results show that the proposed criteria closely follow the performance of optimal, genie-based acceptance.
Future work includes demonstrating these gains in real hardware implementation and developing an \ac{SNR}-independent threshold optimization.

\bibliographystyle{IEEEtran}
\bibliography{references.bib}
\end{NoHyper}
\end{document}